\newif\ifarxiv
\newif\iftcs
\arxivtrue

\iftcs
\documentclass[preprint,12pt]{elsarticle}
\fi

\ifarxiv
\documentclass[a4paper,11pt]{article}
\usepackage[utf8]{inputenc}
\usepackage{a4wide}
\usepackage[numbers,square]{natbib}
\fi

\usepackage{amssymb}
\usepackage{amsmath}
\usepackage{amsthm}

\usepackage[shortlabels]{enumitem}
\usepackage{hyperref}
\usepackage[capitalise]{cleveref}

\iftcs
\journal{Theoretical Computer Science}
\fi

\def\AAtitle{Finitely (In)tractable Promise Constraint Satisfaction Problems}

\def\AAabstract{
The Promise Constraint Satisfaction Problem (PCSP) is a generalization of the Constraint Satisfaction Problem (CSP) that includes approximation variants of satisfiability and graph coloring problems.
Barto [LICS '19]
 has shown that a specific PCSP, the problem
to find a valid Not-All-Equal solution to a 1-in-3-SAT instance, is not finitely tractable in that it can be solved by a trivial reduction to a tractable CSP, but such a CSP is necessarily over an infinite domain (unless P=NP). 
We initiate a systematic study of this phenomenon by giving a general necessary condition for finite tractability and characterizing finite tractability within a class of templates -- the ``basic'' tractable cases in the dichotomy theorem for symmetric Boolean PCSPs allowing negations by Brakensiek and Guruswami [SODA'18].}

\def\AAthanks{An extended abstract of this work appeared in the Proceedings of MFCS 2021~\cite{AB21}. Both authors have received funding from the European Research Council (ERC) under the European Unions Horizon 2020 research and innovation programme (grant agreement No 771005). Kristina Asimi was also funded by  EPSRC grant EP/X03190X/1. Libor Barto was also funded by the European Union (ERC, POCOCOP, 101071674). Views and opinions expressed are however those of the author(s) only and do not necessarily reflect those of the European Union or the European Research Council Executive Agency. Neither the European Union nor the granting authority can be held responsible for them.
}

\begin{document}

\ifarxiv

\title{\AAtitle \thanks{\AAthanks}}

\author{Kristina Asimi\\
Charles University \\
Durham University
\and
Libor Barto\\
Charles University
}

\date{}

\maketitle

\begin{abstract}
    \AAabstract 
\end{abstract}

\fi

\iftcs
\begin{frontmatter}



\title{\AAtitle\tnoteref{thx}}
\tnotetext[thx]{\AAthanks}


\author[label1]{Kristina Asimi\corref{cor1}\fnref{label2}} 
\ead{asimptota94@gmail.com}
\cortext[cor1]{Corresponding author}

\fntext[label2]{Present address: Durham University, UK}
\affiliation[label1]{organization={Department of Algebra, Faculty of Mathematics and Physics, Charles University},
            addressline={Sokolovská 49/83}, 
            city={Prague},
            postcode={186 75}, 
            country={Czechia}}
            
\author[label1]{Libor Barto}
\ead{libor.barto@gmail.com}

\begin{abstract}
\AAabstract
\end{abstract}


\end{frontmatter}

\fi

\theoremstyle{definition}
\theoremstyle{plain}
\newtheorem{theorem}{Theorem}
\newtheorem{lemma}[theorem]{Lemma}
\newtheorem{definition}{Definition}

\newcommand{\CSP}{\mathrm{CSP}}
\newcommand{\PCSP}{\mathrm{PCSP}}
\newcommand{\Pol}{\mathrm{Pol}}
\newcommand{\arity}{\mathrm{ar}}
\newcommand{\rel}[1]{\mathbb{#1}}
\newcommand{\vc}[1]{\mathbf{#1}}
\newcommand{\lrins}{\leq\!r\mbox{-in-}s}
\newcommand{\grins}{\geq\!r\mbox{-in-}s}
\newcommand{\rins}{r\mbox{-in-}s}
\newcommand{\oddins}{\mbox{odd-in-}s}
\newcommand{\evenins}{\mbox{even-in-}s}
\newcommand{\naes}{\mbox{not-all-equal-}s}
\newcommand{\diseq}{(\neq,\neq)}
\newcommand{\ltrins}{\leq\!(2r-1)\mbox{-in-}s}
\newcommand{\gtrins}{\geq\!(2r-s+1)\mbox{-in-}s}
\newcommand{\bra}[1]{\langle #1 \rangle}
\newcommand{\brap}[1]{\langle #1 \rangle_p}
\newcommand{\bran}[1]{\langle #1 \rangle_n}
\newcommand{\tresh}{\theta}
\newcommand{\ts}{t^{\sigma}}
\newcommand{\bbb}{b}
\newcommand{\dz}{\Delta z}
\newcommand{\sgn}{\mathrm{sgn}}
\newcommand{\Ham}{\mathrm{Ham}}
\newcommand{\Par}{\mathrm{Par}}
\newcommand{\Maj}{\mathrm{Maj}}
\newcommand{\AT}{\mathrm{AT}}


\section{Introduction}
\label{intro}

Many computational problems, including various versions of logical satisfiability, graph coloring, and systems of equations can be phrased as Constraint Satisfaction Problems (CSPs) over fixed templates (see~\cite{BKW17}). One of the possible formulations of the CSP is via homomorphisms of relational structures: a \emph{template} $\rel A$ is a relational structure with finitely many relations and the CSP over $\rel A$, written $\CSP(\rel A)$, is the problem to decide whether a given finite relational structure  $\rel X$ (similar to $\rel A$) admits a homomorphism to $\rel A$. 

The complexity of CSPs over finite templates (i.e., those templates whose domain is a finite set) is now completely classified by a celebrated dichotomy theorem independently obtained by Bulatov~\cite{Bul17} and Zhuk~\cite{Zhu17,Zhu20}: every $\CSP(\rel A)$ is either tractable (that is, solvable in polynomial-time) or NP-complete. The landmark results leading to the complete classification include Schaefer's dichotomy theorem~\cite{Sch78} for CSPs over Boolean structures (i.e., structures with a two-element domain), Hell and Ne\v set\v ril's dichotomy theorem~\cite{HN90} for CSPs over graphs, and Feder and Vardi's thorough study~\cite{FV98} through Datalog and group theory. The latter paper also inspired the development of a mathematical theory of finite-template CSPs~\cite{Jea98, BJK05,BOP18}, the so called \emph{algebraic approach}, that provided guidance and tools for the general dichotomy theorem by Bulatov and Zhuk. 

The algebraic approach has been successfully applied in many variants and generalizations of the CSP such as the infinite-template CSP~\cite{Bod21} 
or
valued CSP~\cite{KKR17}. This paper concerns a recent vast generalization of the basic CSP framework, the Promise CSP (PCSP). 

A \emph{template} for the PCSP is a pair $(\rel A,\rel B)$ of similar structures such that $\rel A$ has a homomorphism to $\rel B$, and the PCSP over $(\rel A,\rel  B)$, written $\PCSP(\rel A,\rel B)$, is the problem to distinguish between the case that a given finite structure $\rel X$ admits a homomorphism to $\rel A$ and the case that $\rel X$ does not have a homomorphism to $\rel B$ (the promise is that one of the cases takes place). This framework generalizes that of CSP (take $\rel A = \rel B$) and additionally includes important problems in approximation, e.g., if $\rel A = \rel K_{k}$ (the clique on $k$ vertices) and $\rel B = \rel K_l$, $k \leq l$, then $\PCSP(\rel A,\rel B)$ is a version of the approximate graph coloring problem, namely, the problem to distinguish graphs that are $k$-colorable from those that are not $l$-colorable, a problem whose complexity is open after more than 40 years of research. On the other hand, the basics of the algebraic approach to CSPs can be generalized to PCSPs~\cite{AGH17,BG21,BBKO21,KO22}.

The approximate graph coloring problem shows that a full classification of the complexity of PCSPs over graph templates is still open and so is the analogue of Schaefer's Boolean CSP, PCSPs over pairs of Boolean structures. However, strong partial results have already been obtained. Brakensiek and Guruswami~\cite{BG21} proved
a dichotomy theorem for all symmetric Boolean templates allowing negations, i.e., templates $(\rel A, \rel B)$ such that $\rel{A}=(\{0,1\}; R_0, R_1, \dots)$, $\rel{B} = (\{0,1\}; S_0, S_1, \dots)$, each relation $R_i$, $S_i$ is invariant under permutations of coordinates, and $R_0=S_0$ is the binary disequality relation $\neq$. These templates play a central role in this paper. Ficak, Kozik, Ol\v s\'ak, and Stankiewicz~\cite{FKOS19} later generalized this result to all symmetric Boolean templates.

To prove tractability or hardness results for PCSPs, a very simple but useful reduction is often applied: If $(\rel A,\rel B)$ and $(\rel A',\rel B')$ are similar PCSP templates and there exist homomorphisms $\rel A' \to \rel A$ and $\rel B \to \rel B'$, then the trivial reduction (which does not change the instance) reduces $\PCSP(\rel A',\rel B')$ to $\PCSP(\rel A,\rel B)$; we say that $(\rel A',\rel B')$ is a \emph{homomorphic relaxation} of $(\rel A,\rel B)$. In fact, all the tractable symmetric Boolean PCSPs can be reduced in this way to a tractable CSP over a structure with a \emph{possibly infinite domain}. 

An interesting example of a PCSP that can be naturally reduced to a tractable CSP over an infinite domain is the following problem. An instance is a list of triples of variables and the problem is to distinguish instances that are satisfiable as positive 1-in-3-SAT instances from those that are not even satisfiable as Not-All-Equal-3-SAT instances. This computational problem is essentially the same as $\PCSP(\rel{A},\rel{B})$ where $\rel{A}$ consists of the ternary 1-in-3 relation over $\{0,1\}$  and $\rel B$ consists of the ternary not-all-equal relation over $\{0,1\}$. It is easy to see that $\rel A \to \rel C \to \rel B$
where $\rel{C}$ is the relation ``$x+y+z=1$'' over the set of all integers. Therefore $\PCSP(\rel A, \rel B)$ is reducible (by means of the trivial reduction) to $\PCSP(\rel C, \rel C) = \CSP(\rel C)$ which is a tractable problem. The main result of~\cite{Bar19} (Section 8 in \cite{BBKO21}) is that no finite structure can be used in place of $\rel C$ for this particular template -- this PCSP is not finitely tractable in the sense of the following definition.

\begin{definition}
We say that $\PCSP({\rel{A},\rel{B}})$ is \emph{finitely tractable} if there exists a finite relational structure $\rel{C}$ such that $\rel{A} \to \rel{C} \to \rel{B}$ and $\CSP({\rel{C}})$ is tractable. Otherwise we call $\PCSP(\rel A,\rel B)$ \emph{not finitely tractable}, or \emph{finitely intractable}. (We assume P $\neq$ NP throughout the paper.)
\end{definition}

In this paper, we initiate a systematic study of this phenomenon. As the main technical contribution, we determine which of the ``basic tractable cases'' in Brakensiek and Guruswami's classification~\cite{BG21} are finitely tractable. It turns out that finite tractability is quite rare, so the infinite nature of the 1-in-3 versus Not-All-Equal problem is not exceptional at all.
On the other hand, there are interesting examples of finitely tractable PCSPs  \cite{NZ24,LZ25,M24}.

\subsection{Symmetric Boolean PCSPs allowing negations}\label{subsection:basiccases}

We now discuss the classification of symmetric Boolean templates allowing negations from~\cite{BG21}. It will be convenient to describe these templates by listing the corresponding relation pairs, that is, instead of $(\rel A = (\{0,1\}; R_1, R_2, \dots, R_n), \rel B = (\{0,1\}; S_1, S_2, \dots, S_n))$ we describe this template by the list $(R_1,S_1)$, $(R_2,S_2)$, \dots, $(R_n,S_n)$. Recall that the template is \emph{symmetric} if all the involved relations are symmetric, i.e., invariant under any permutation of coordinates, and the template \emph{allows negations} if $\diseq$ is among the relation pairs, where $\neq$, defined as $\{(0,1),(1,0)\}$, is the disequality relation.   

It may be also helpful to think of an instance of $\PCSP(\rel A,\rel B)$ as a list of constraints of the form $R_i(\mbox{variables})$ and the problem is to distinguish between instances where each constraint is satisfiable and those which are not satisfiable even when we replace each $R_i$ by the corresponding ``relaxed version'' $S_i$. Allowing negations then means that we can use constraints $x \neq y$ -- we can effectively negate variables.  

The following relations are important for the classification.  

\begin{itemize}
   \item $\oddins = \{\vc{x}\in\{0,1\}^s:\sum_{i=1}^{s} x_i \mbox{ is odd}\}$, \\ 
   $\evenins = \{\vc{x} \in \{0,1\}^s:\sum_{n=1}^{s} x_i \mbox{ is even}\}$ 
   \item 
   $\rins = \{\vc{x} \in \{0,1\}^s:\sum_{n=1}^{s} x_i = r\}$
   \item 
   $\lrins = \{\vc{x}\in\{0,1\}^s:\sum_{i=1}^{s} x_i\leq r\}$, \\ 
   $\grins = \{\vc{x}\in\{0,1\}^s:\sum_{i=1}^{s} x_i\geq r\}$
   \item
   $\naes = \{\vc{x} \in \{0,1\}^s: \sum_{i=1}^{s} x_i \not\in \{0,s\}\}$
\end{itemize}

The next theorem lists some of the tractable cases of the classification, which are ``basic'' in the sense explained below. 

\begin{theorem}[\cite{BG21}] \label{thm:bg} 
$\PCSP((P, Q),\diseq)$ is tractable if $(P,Q)$ is equal to
\begin{enumerate}[(a)]
\item $(\oddins,\oddins)$, or $(\evenins,\evenins)$, or
\item $(\lrins,\ltrins)$ and $r \leq s/2$, or \\ 
      $(\grins, \gtrins)$ and $r \geq s/2$, or   
\item $(\rins,\naes)$

\end{enumerate}
for some positive integers $r,s$.
\end{theorem}

It can be derived from the results in \cite{BG21} (see \cref{append}) that every tractable symmetric Boolean PCSP allowing negations can be obtained by
\begin{itemize}
    \item taking 
    any number of relation pairs from one of the following three items (where $r$ and $s$ are positive integers):
    \begin{enumerate}[(a)]
\item $(\oddins,\oddins)$, or $(\evenins,\evenins)$
\item $(\lrins,\ltrins)$ and $r \leq s/2$, or \\ 
      $(\grins, \gtrins)$ and $r \geq s/2$, or \\
      $(\frac{s}{2}\mbox{-in-}s,\naes)$ and $s$ is even
\item $(\rins,\naes)$
\end{enumerate} 
    \item adding any number of ``trivial'' relation pairs $(P,Q)$ such that $P \subseteq Q$, and $Q$ is the full relation or $P$ contains only constant tuples, and 
    \item taking a homomorphic relaxation of the obtained template.
\end{itemize}
In this sense, \cref{thm:bg} provides building blocks for all tractable templates.

\subsection{Contributions}

Some of the cases in \cref{thm:bg} are finitely tractable: templates in item (a) are tractable CSPs (they can be decided by solving systems of linear equations of the two-element field), templates in item (c) for $r$ odd and $s$ even are homomorphic relaxations of $(\oddins,\oddins)$, and templates in item (b) for $r=1$ or $r=s-1$ as well as all templates with $s \leq 2$ are tractable CSPs (reducible to 2-SAT)~\cite{Sch78,BKW17}. Our main theorem proves that all the remaining cases are not finitely tractable. In fact, we prove this property even for some relaxations of these templates:

\begin{theorem} \label{thm:main}
\sloppy The PCSP over any of the following templates is not finitely 
tractable. 
\begin{itemize}
    \item[(1)] $(\rins,\ltrins), \diseq$ where $1 < r < s/2$, \\
          $(\rins,\gtrins), \diseq$ where $s/2 < r < s-1$ 
    \item[(2)] $(\lrins,\ltrins), \diseq$ where $s$ is even, $1 < r = s/2$\\
          $(\grins,\gtrins), \diseq$ where $s$ is even, $1 < r = s/2$
    \item[(3)] $(\rins,\ltrins), \diseq$ where $s$ is even, $1 < r = s/2$, and $r$ is even \\
          $(\rins,\gtrins), \diseq$ where $s$ is even, $1 < r = s/2$, and $r$ is even    
   \item[(4)] $(\rins,\naes)$ where $s>r$, $s>2$, and $r$ is even or $s$ is odd
\end{itemize}
\end{theorem}

Note that the templates in the last item do not contain the disequality pair; the special case with $r=1$ and $s=3$ is the main result of \cite{Bar19}. Disequalities in the other items are necessary, since otherwise the templates are homomorphic relaxations of CSPs over one-element structures. 

In \cref{thm:crit} we provide a general necessary condition for finite tractability of an arbitrary finite-template PCSP in terms of so called h1 identities. 
Showing that templates in \cref{thm:main} do not satisfy this necessary condition forms the bulk of the paper. 

The necessary condition in \cref{thm:crit} seems very unlikely to be sufficient for finite tractability. Nevertheless, we observe in \cref{thm:h1} that finite tractability \emph{does} depend only on h1 identities, just like polynomial-time solvability~\cite{BBKO21}, see \cref{thm:minion} and the discussion following the theorem.

\section{Preliminaries}

\subsection{PCSP}

For every positive integer $n$ we let $[n] = \{1, 2, \dots, n\}$.

A \emph{relational structure} (of finite signature) is a tuple $\rel{A}=(A;R_1,R_2,\dots,R_n)$ where $A$ is a set, called the \emph{domain}, and each $R_i$ is a relation on $A$ of arity $\arity(R_i)\geq 1$, that is, $R_i \subseteq A^{\arity(R_i)}$. The structure $\rel{A}$ is finite if $A$ is finite. Two relational structures $\rel{A}=(A;R_1,R_2,\dots,R_n)$ and $\rel{B}=(B;S_1,S_2,\dots,S_n)$ are \emph{similar} if they have the same number of relations and $\arity(R_i)=\arity(S_i)$ for each $i\in [n]$.
In this case, a \emph{homomorphism} from $\rel{A}$ to $\rel{B}$ is a mapping $f:A\rightarrow B$ such that $(f(a_1),f(a_2),\dots,f(a_k))\in S_i$ whenever $i\in [n]$ and $(a_1,a_2,\dots,a_k)\in R_i$ where $k=\arity(R_i)$.
If there exists a homomorphism from $\rel{A}$ to $\rel{B}$, we write $\rel{A}\rightarrow\rel{B}$, and if there is none, we write $\rel{A}\not\to\rel{B}$.

\begin{definition}
A \emph{$\PCSP$ template} is a pair $(\rel{A},\rel{B})$ of similar relational structures such that $\rel{A}\rightarrow\rel{B}$.

The $\PCSP$ over $(\rel{A},\rel{B})$, written $\PCSP(\rel{A},\rel{B})$, is the following problem. Given a finite relational structure $\rel{X}$ similar to $\rel{A}$ (and $\rel{B}$), output ``Yes.'' if $\rel{X} \to \rel{A}$ and output ``No.'' if $\rel{X} \not\to \rel{B}$. 

We define $\CSP(\rel{A})=\PCSP(\rel{A},\rel{A})$.
\end{definition}

\begin{definition}
Let $(\rel{A},\rel{B})$ and $(\rel{A'},\rel{B'})$ be similar $\PCSP$ templates. We say that $(\rel{A'},\rel{B'})$ is a \emph{homomorphic relaxation} of $(\rel{A},\rel{B})$ if $\rel{A'}\rightarrow\rel{A}$ and $\rel{B}\rightarrow\rel{B'}$.
\end{definition}

Recall that if $(\rel{A'},\rel{B'})$ is a homomorphic relaxation of $(\rel{A},\rel{B})$, then the trivial reduction, which does not change the input structure $\rel{X}$, reduces $\PCSP(\rel{A'},\rel{B'})$ to $\PCSP(\rel{A},\rel{B})$.

\subsection{Polymorphisms}

A crucial concept for the algebraic approach to (P)CSP is a polymorphism.

\begin{definition}
    Let $R\subseteq A^k$ and $S\subseteq B^k$ be relations. A function $c:A^n\to B$ is a \emph{polymorphism} of $(R,S)$ if 
    \[
\begin{pmatrix}
a_{11}\\
a_{21}\\
\vdots\\
a_{k1}
\end{pmatrix}\in R,
\begin{pmatrix}
a_{12}\\
a_{22}\\
\vdots\\
a_{k2}
\end{pmatrix}\in R,
\ldots,
\begin{pmatrix}
a_{1n}\\
a_{2n}\\
\vdots\\
a_{kn}
\end{pmatrix}\in R\Rightarrow
\begin{pmatrix}
c(a_{11},a_{12},\dots,a_{1n})\\
c(a_{21},a_{22},\dots,a_{23})\\
\vdots\\
c(a_{k1},a_{k2},\dots,a_{kn})
\end{pmatrix}\in S.
\]
\end{definition}

\begin{definition} Let $\rel{A}=(A;R_1,R_2,\dots,R_m)$ and $\rel{B}=(B;S_1,S_2,\dots,S_m)$ be two similar relational structures. A function $c:A^n\rightarrow B$ is a \emph{polymorphism} from $\rel{A}$ to $\rel{B}$ if it is a polymorphism of $(R_i,S_i)$ for every $i\in\{1,2,\dots,m\}.$

We denote the set of all polymorphisms from $\rel{A}$ to $\rel{B}$ by $\Pol(\rel{A},\rel{B})$ and define $\Pol(\rel C) = \Pol(\rel C, \rel C)$.
\end{definition}

The computational complexity of a PCSP depends only on the set of polymorphisms of its template~\cite{BG21}. We note that tractability of the PCSPs in \cref{thm:bg} stems from nice polymorphisms: parities (item~(a)), majorities (item~(b)), and alternating thresholds (item~(c)). 

The set of polymorphisms is an algebraic object named minion in \cite{BBKO21}, which we define in \cref{def:minion} below.

\begin{definition}
An $n$-ary function $f^\pi : A^n \to B$ is called a \emph{minor} of an $m$-ary function $f:A^m \to B$ given by a map $\pi:[m]\to [n]$ if
\[ f^\pi(x_1,x_2,\dots,x_n) =f(x_{\pi(1)},x_{\pi(2)}\dots,x_{\pi(m)})\] for all $x_1,x_2,\dots,x_n \in A$.
\end{definition}

\begin{definition} \label{def:minion}
Let $\mathcal{O}(A,B)=\{f:A^n\to B:n\geq 1\}$. A \emph{minion} on $(A,B)$ is a non-empty subset $\mathcal{M}$ of $\mathcal{O}(A,B)$ that is closed under taking minors. For fixed $n\geq 1$, let $\mathcal{M}^{(n)}$ denote the set of $n$-ary functions from $\mathcal{M}$.
\end{definition}

As mentioned, $\mathcal{M} = \Pol(\rel A,\rel B)$ is always a minion and the complexity of $\PCSP(\rel A,\rel B)$ depends only on $\mathcal{M}$. This result was strengthened in~\cite{BBKO21} (generalizing the same result for CSPs~\cite{BOP18}) as follows.

\begin{definition}
Let $\mathcal{M}$ and $\mathcal{N}$ be two minions. A mapping $\xi: \mathcal{M} \to \mathcal{N}$ is called a \emph{minion homomorphism} if it preserves arities and preserves taking minors, i.e., $\xi(f^{\pi}) = (\xi(f))^{\pi}$ for every $f \in \mathcal{M}^{(m)}$ and every $\pi: [m] \to [n]$.
\end{definition}

\begin{theorem} \label{thm:minion}
Let $(\rel A, \rel B)$ and $(\rel A', \rel B')$ be PCSP templates. If there exists a minion homomorphism $\Pol(\rel A', \rel B'
) \to \Pol(\rel A,\rel B)$, then $\PCSP(\rel A,\rel B)$ is log-space reducible to $\PCSP(\rel A',\rel B')$. 
\end{theorem}

An \emph{h1 identity} (where h1 stands for height one) is a meaningful  expression of the form $\mbox{function}(\mbox{variables}) \approx \mbox{function}(\mbox{variables})$, e.g., if $f: A^3 \to B$ and $g: A^4 \to B$, then $f(x,y,x) \approx g(y,x,x,z)$ is an h1 identity. Such an h1 identity is \emph{satisfied} if the corresponding equation holds universally, e.g., $f(x,y,x) \approx g(y,x,x,z)$ is satisfied if and only if $f(x,y,x) = g(y,x,x,z)$ for every $x,y,z \in A$. 

Every minion homomorphism $\xi: \mathcal{M} \to \mathcal{N}$ preserves h1 identities in the sense that if functions $f,g \in \mathcal{M}$ satisfy an h1 identity, then so do their $\xi$-images $\xi(f),\xi(g) \in \mathcal{N}$. In fact, an arity-preserving $\xi$ between minions is a minion homomorphism if and only if it preserves h1 identities (see~\cite{BOP18} for details). In this sense, \cref{thm:minion} shows that the complexity of a PCSP depends only on h1 identities satisfied by polymorphisms.

\subsection{Notation for tuples} \label{subsec:notation}

Repeated entries in tuples will be indicated by $\times$, e.g. $(2 \times a, 3 \times b)$ stands for the tuple $(a,a,b,b,b)$.

The $i$-th \emph{cyclic shift} of a tuple $(x_1, \dots, x_m)$ is the tuple  
\[(x_{(m-i \mod m)+1}, \dots, x_m,x_1, \dots, x_{(m-i-1 \mod m)+1}).
\]
A \emph{cyclic shift} is the $i$-th cyclic shift for some $i$. We will use cyclic shifts both for tuples of zeros and ones and tuples of variables. 

We will often use special $p$-tuples and $n=p^2$-tuples of zeros and ones as arguments for Boolean functions, where $p$ will be a fixed prime number.   For $0 \leq k \leq p$, $0 \leq l \leq p^2$, and $0 \leq k^1, \dots, k^p \leq p$ we write
\[
\brap{k} = (k \times 1, (p-k) \times 0) =  (\underbrace{1,1, \dots, 1}_{k}, \underbrace{0,0, \dots,0}_{p-k}), \quad
\bran{l} = (\underbrace{1,1, \dots, 1}_{l}, \underbrace{0,0, \dots,0}_{n-l})
\]
and 
\[
\brap{k^1, k^2, \dots, k^p} = \brap{k^1}\brap{k^2}\dots\brap{k^p}
\]
for the concatenation of $\brap{k^1}$, $\brap{k^2}$\dots, $\brap{k^p}$. (Note here that the ``$i$'' in $k^i$ is an index, not an exponent.) The subscripts $p$ and $n$ in $\brap{}$ and $\bran{}$ will be usually clear from the context and we omit them.
We will sometimes need to shift $n$-ary tuples $\bra{k^1, k^2, \dots, k^p}$ blockwise, e.g., to $\bra{k^2, \dots, k^p,k^1}$. In such a situation we talk about a \emph{$p$-ary cyclic shift} to avoid confusion. 

It will be often convenient to think of an $n$-tuple $\vc{k} = \bra{k^1, k^2, \dots, k^p}$ as a $p \times p$ zero-one matrix with columns $\bra{k^1}$, $\bra{k^2}$,\dots, $\bra{k^p}$. For example,  $\bra{p \times 3}$ is the $p \times p$ matrix whose ones form a $3 \times p$ ``rectangle''. As another example, the ones in $\bra{(p-2) \times 3, 2 \times 2}$ form ``almost'' a $3 \times p$ rectangle -- the bottom right $1 \times 2$ corner is removed. For $p=5$ we have the following matrix.
\[
\left(\begin{array}{ccccc} 1 & 1 & 1 & 1 & 1 \\ 1 & 1 & 1 & 1 & 1 \\
1 & 1 & 1 & 0 & 0 \\ 0 & 0 & 0 & 0 &0\\ 0 & 0 & 0 & 0 & 0 
\end{array}\right)
\]
A $p$-ary cyclic shift of $\vc{k}$ corresponds to a cyclic permutation of columns.

The area of a zero-one $n$-tuple $\vc{k}$ is defined as the fraction of ones and is denoted $\lambda(\vc{k})$. 
\[
\lambda(\vc{k}) = \left(\sum_{i=1}^n k_i\right) / p^2
\]
The area of $\bra{k^1, k^2, \dots, k^p}$ is thus $(k^1+k^2+\dots+k^p)/p^2$.

If $t$ is a $p$-ary function, we simply write $t \bra{k}$  instead of $t(\bra{k})$. Similar shorthand is used for $n$-ary functions and tuples  $\brap{k^1,k^2, \dots, k^p}$.

\section{Finitely tractable PCSPs}

\subsection{Finite tractability depends only on h1 identities}

We start by observing that finite tractability also depends only on h1 identities satisfied by polymorphisms, just like standard tractability (recall the discussion about h1 identities and minion homomorphisms below \cref{thm:minion}). This result, \cref{thm:h1}, is an immediate consequence of the following lemma and \cref{thm:minion}.  

\begin{lemma} \label{lem:FTandMinions}
Let $(\rel A,\rel B)$ be a PCSP template. Then the following are equivalent.
\begin{itemize}
    \item $\PCSP(\rel A, \rel B)$ is finitely tractable.
    \item There exists a finite relational structure $\rel{C}$ such that $\CSP(\rel{C})$ is solvable in polynomial time and there exists a minion homomorphism $\Pol(\rel{C}) \to \Pol(\rel A, \rel B)$.
\end{itemize}
\end{lemma}

\begin{proof}
This lemma is a consequence of known results and we only sketch the argument here. In Section II.B of~\cite{Bar19} it is argued that the first item is equivalent to the claim that a finite tractable template $(\rel C,\rel C)$ pp-constructs $(\rel A,\rel B)$. The latter claim is equivalent to the second item by Theorem 4.12 in~\cite{BBKO21}.
\end{proof}

\begin{theorem} \label{thm:h1}
Let $(\rel A, \rel B)$ and $(\rel A', \rel B')$ be PCSP templates. If there exists a minion homomorphism $\Pol(\rel A', \rel B'
) \to \Pol(\rel A,\rel B)$ and $\PCSP(\rel A',\rel B')$ is finitely tractable, then so is $\PCSP(\rel A,\rel B)$.
\end{theorem}

\subsection{Necessary condition for finite tractability}

In this subsection, we derive the necessary condition for finite tractability that will be used to prove \cref{thm:main}. A cyclic polymorphism is a starting point for the condition.

\begin{definition}
A function $c:A^p\rightarrow B$ is called \emph{cyclic} if it satisfies the h1 identity
\begin{center}
$c(x_1,x_2,\dots,x_p) \approx c(x_2,\dots ,x_p,x_1)$.
\end{center}
\end{definition}

Cyclic polymorphisms can be used~\cite{BK12} to characterize the borderline between tractable and NP-complete CSPs proposed in~\cite{BJK05} and confirmed in~\cite{Bul17,Zhu17,Zhu20}. We only state the direction needed in this paper.

\begin{theorem}[\cite{BK12}]\label{thm:cyclic}
Let $\rel{C}$ be a CSP template over a finite domain $C$. If $\mathrm{CSP}(\rel{C})$ is not NP-complete, then $\rel{C}$ has a cyclic polymorphism of arity $p$ for every prime number $p>|C|$. 
\end{theorem}

Polymorphism minions of CSP templates are closed under arbitrary composition~(cf. \cite{BKW17}). In particular, if $\CSP(\rel{C})$ is not NP-complete, then $\Pol(\rel{C})$ contains the function
\begin{align}\label{eq:t}
\begin{split}
t&(x_{11},x_{21}, \dots, x_{p1}, \ x_{12}, x_{22}, \dots, x_{p2}, \ \dots, \ x_{1p}, x_{2p}, \dots, x_{pp}) \\
&= 
c(c(x_{11}, x_{21}, \dots, x_{p1}), c(x_{12}, x_{22}, \dots, x_{p2}), \dots, c(x_{1p}, x_{2p}, \dots, x_{pp})),
\end{split}
\end{align}
where $c$ is a $p$-ary cyclic function and $p > |C|$.
Such a function satisfies strong h1 identities which are not satisfied by the templates in \cref{thm:main}.  We now (in two steps) describe one such collection of strong enough identities.

\begin{definition} \label{def:doublycyclic}
A function $t: A^{p^2} \to B$ is \emph{doubly cyclic} if it satisfies every identity of the form $t(\vc{x}_1, \vc{x}_2,\dots, \vc{x}_p) \approx t(\vc{y}_1,\vc{y}_2, \dots, \vc{y}_p)$, where $\vc{x}_i$ is a $p$-tuple of variables and $\vc{y}_i$ is a cyclic shift of $\vc{x}_i$ for every $i \in [p]$, and every identity of the form $t(\vc{x}_1, \vc{x}_2 \dots, \vc{x}_p) \approx t(\vc{x}_2, \dots, \vc{x}_p, \vc{x}_1)$, where each $\vc{x}_i$ is a $p$-tuple of variables. 
\end{definition}

Observe that $t$ from \cref{eq:t} is doubly cyclic -- the first type of identities come from the cyclicity of the inner $c$ while the second type from the outer $c$. 
It will be also useful for us to observe in \cref{tscyclic} that, after rearranging the arguments (we read them row-wise), $t$ is a cyclic function of arity $p^2$. 
From the finiteness of the domain $C$ we get one more property of function $t$.
In the next definition, by an \emph{$x/y$-tuple} we mean a tuple containing only variables $x$ and $y$.

\begin{definition} \label{def:bounded}
A doubly cyclic function $t: A^{p^2} \to B$ is \emph{$\bbb$-bounded} if there exists an equivalence relation $\sim$ on the set of all $p$-ary $x/y$-tuples with at most $\bbb$ equivalence classes such that $t$ satisfies every  identity of the form
$t(\vc{u}_1,\vc{u}_2, \dots \vc{u}_p) \approx t(\vc{v}_1,\vc{v}_2, \dots, \vc{v}_p)$
where $\vc{u}_i$ and $\vc{v}_i$ are $x/y$-tuples such that $\vc{u}_i \sim \vc{v}_i$ for every $i \in [p]$.
\end{definition}

\begin{lemma} \label{lem:bounded}
Let $c: C^p \to C$ be a cyclic function. Then the function $t$ defined by \cref{eq:t} is a $\bbb$-bounded doubly cyclic function for $\bbb = |C|^{|C|^2}$. 
\end{lemma}

\begin{proof}
We define $\sim$ by declaring two $p$-ary $x/y$-tuples $\vc{u}$ and $\vc{v}$ $\sim$-equivalent if $c(\vc{u}) \approx c(\vc{v})$. As there are $\bbb=|C|^{|C|^2}$ binary functions $C^2 \to C$, this equivalence has at most $\bbb$ equivalence classes. By definitions, $t$ is then $\bbb$-bounded and doubly cyclic. 
\end{proof}

The promised necessary condition for finite tractability is now a simple consequence:

\begin{theorem} \label{thm:crit}
Let $(\rel A, \rel B)$ be a finite PCSP template that is finitely tractable.
Then there exists $\bbb$ such that $(\rel A, \rel B)$ has  a $p^2$-ary $\bbb$-bounded doubly cyclic polymorphism for every sufficiently large prime $p$. 
\end{theorem}

\begin{proof}
If $(\rel A, \rel B)$ is finitely tractable, then, by \cref{lem:FTandMinions}, there exists a minion homomorphism $\xi: \Pol(\rel C) \to \Pol(\rel A,\rel B)$, where $\rel C$ is finite and $\CSP(\rel C)$ is tractable. By \cref{thm:cyclic}, $\rel C$ has a $p$-ary cyclic polymorphism for every sufficiently large prime. Then, by \cref{lem:bounded}, the polymorphism $t$ of $\rel C$ defined by~\cref{eq:t} is a $\bbb$-bounded and doubly cyclic (with the appropriate $\bbb$). As $\xi$ preserves h1 identities, $\xi(t)$ is a $\bbb$-bounded doubly cyclic polymorphism of $(\rel A,\rel B)$. 
\end{proof}

\section{Proof of the main result}
In this section we prove \cref{thm:main}.  Without loss of generality, we consider only templates on the first lines of Cases (1)--(3) of \cref{thm:main} (in particular, $r \leq s/2$) and assume that $r \leq s/2$ in Case (4). The remaining templates can be obtained by swapping zero and one in the domains. Therefore we have the following cases.

\begin{description}
    \item[Case (1)]  $\PCSP((\rins,\ltrins), \diseq)$ where $1 < r < s/2$
    \item[Case (2)] $\PCSP((\lrins,\ltrins), \diseq)$ where $s$ is even, $1 < r = s/2$
    \item[Case (3)]
    $\PCSP((\rins,\ltrins), \diseq)$ where $s$ is even, $1 < r = s/2$, and $r$ is even
    \item[Case (4)] $\PCSP(\rins,\naes)$ where $r \leq s/2$, $s>2$, and $r$ is even or $s$ is odd
\end{description}

The cases will be treated simultaneously and we only distinguish them where necessary. The first relation pair in the template is denoted $(P,Q)$. 

 Striving for a contradiction, suppose that the PCSP is finitely tractable. By \cref{thm:crit} there exists $\bbb$ such that the template has a $p^2$-ary $\bbb$-bounded doubly cyclic polymorphism $t$ for every sufficiently large arity $p^2$. 
We fix such a $\bbb$ and $t$, where $p$ is fixed to a sufficiently large prime $p$ congruent  to 1 modulo $s$ (which is possible by the Dirichlet prime number theorem). How large must $p$ be will be seen in due course. We denote $n=p^2$ and observe that $n \equiv 1 \pmod{s}$ as well. 

Using the cyclicity and double cyclicity we will show that certain $n$-tuples $\vc{z}$ are tame in that $t(\vc{z}) = t \bra{0}$ (recall here the notation in \cref{subsec:notation}) iff the area of $\vc{z}$ is below a threshold $\tresh$. The formal definitions are as follows.

\begin{definition}
The \emph{threshold} $\tresh$ is defined as follows.
\begin{description}
    \item[Case (1),(2),(3)] $\tresh = 1/2$ 
    \item[Case (4)] $\tresh = r/s$
\end{description}
(Observe that $\tresh=r/s$ also in Case (2) and (3).)

A tuple $\vc{z} \in \{0,1\}^n$ is \emph{tame} if
\[
 t(\vc{z}) = \left\{ \begin{array}{ll} 
 t \bran{0} & \mbox{if $\lambda(\vc{z}) < \tresh$} \\ 
 1-t \bran{0} & \mbox{if $\lambda(\vc{z}) > \tresh$} \\ 
 \end{array}\right.
\]
(Note here that $\lambda(\vc{z})$ is never equal to $\tresh$ since $n$ is odd and $n \equiv 1 \pmod{s}$.)
\end{definition}

The evaluations that we use are called near-threshold almost rectangles defined as follows.

\begin{definition}
A tuple $\vc{z} \in \{0,1\}^n$ is an \emph{almost rectangle} if it is a $p$-ary cyclic shift of a tuple of the form $\brap{z^1, \dots,z^1, z^2, \dots, z^2}$, where $0 \leq z^1,z^2 \leq p$, the number of $z^1$'s is arbitrary, and $|z^1 - z^2| < 5\bbb$. The quantity $\dz=|z^1-z^2|$ is called \emph{step size}. We say that $\vc{z}$ is \emph{near-threshold} if $|\lambda(\vc{z})-\tresh| < 1/s^{\dz+3}$. 
\end{definition}

In \cref{subsec:stepone} and \cref{subsec:stepx} we show that near-threshold almost rectangles are tame. The proof will then be finished in \cref{subsec:contra} by finding two near-threshold almost rectangles $\vc{z}_1$ and $\vc{z}_2$ such that $\lambda(\vc{z}_1) < \tresh < \lambda(\vc{z}_2)$ and $t(\vc{z}_1) = t(\vc{z}_2)$, which will contradict the tameness.

\subsection{Step size one} \label{subsec:stepone}

In this subsection we work towards \cref{lem:stepone} which shows that near-threshold almost rectangles of step size at most one are tame. This will provide us with the base case for an inductive proof in the next subsection.

An almost rectangle of step at most one has a $p$-ary cyclic shift of the form $\vc{z} = \brap{z^2+1, \dots, z^2+1,z^2, \dots, z^2}$. Observe that this tuple regarded as a $p \times p$ matrix is, when read row-wise, equal to a sequence of consecutive ones, followed by zeros. Therefore the following lemma will be useful. It shows that the function $\ts$ obtained from $t$ by swapping the roles of rows and columns is cyclic.

\begin{lemma} \label{tscyclic}
Let $t: A^{p^2} \to B$ be a doubly cyclic function. Then the function $\ts$ defined by
\[\ts \left(
\begin{array}{cccc}
x_{11} & x_{12} & \cdots & x_{1p} \\
x_{21} & x_{22} & \cdots & x_{2p} \\
\vdots & \vdots & \ddots & \vdots \\
x_{p1} & x_{p2} & \cdots & x_{pp}
\end{array}
\right) = t\left(
\begin{array}{cccc}
x_{11} & x_{21} & \cdots & x_{p1} \\
x_{12} & x_{22} & \cdots & x_{p2} \\
\vdots & \vdots & \ddots & \vdots \\
x_{1p} & x_{2p} & \cdots & x_{pp}
\end{array}
\right)\]
is a cyclic function.
\end{lemma}

\begin{proof} By cyclically shifting the arguments we get the same result:
\begin{align*}
\begin{split}
 \ts &(x_{21},x_{31},\dots,x_{p1},x_{12},x_{22},x_{32},\dots,x_{p2},x_{13},\dots,x_{2p},x_{3p},\dots,x_{pp},x_{11}) \\
 & =\ts \left(
\begin{array}{cccc}
x_{21} & \cdots & x_{2,p-1} & x_{2p} \\
\vdots & \ddots & \vdots & \vdots \\
x_{p1} & \cdots & x_{p,p-1} & x_{pp}\\
x_{12} & \cdots & x_{1p} & x_{11}
\end{array}
\right)=t \left(
\begin{array}{cccc}
x_{21} & \cdots & x_{p1} & x_{12} \\
\vdots & \ddots & \vdots & \vdots \\
x_{2,p-1} & \cdots & x_{p,p-1} & x_{1p}\\
x_{2p} & \cdots & x_{pp} & x_{11}
\end{array}
\right)\\
 &=t\left(
\begin{array}{cccc}
x_{21} & \cdots & x_{p1} & x_{11} \\
\vdots & \ddots & \vdots & \vdots \\
x_{2,p-1} & \cdots & x_{p,p-1} & x_{1,p-1}\\
x_{2p} & \cdots & x_{pp} & x_{1p}
\end{array}
\right)=t\left(
\begin{array}{cccc}
x_{11} & x_{21} & \cdots & x_{p1} \\
x_{12} & x_{22} & \cdots & x_{p2} \\
\vdots & \vdots & \ddots & \vdots \\
x_{1p} & x_{2p} & \cdots & x_{pp}
\end{array}
\right)\\
&=\ts \left(
\begin{array}{cccc}
x_{11} & x_{12} & \cdots & x_{1p} \\
x_{21} & x_{22} & \cdots & x_{2p} \\
\vdots & \vdots & \ddots & \vdots \\
x_{p1} & x_{p2} & \cdots & x_{pp}
\end{array}
\right)\\
&=\ts(x_{11},x_{21},\dots,x_{p1},x_{12},x_{22},\dots,x_{p2},\dots,x_{1p},x_{2p},\dots,x_{pp}).
\end{split}
\end{align*}
\end{proof}
 
The following lemma is a consequence of the fact that $\ts$ is a polymorphism (as $t$ is) which is, as we have just shown, cyclic.

\begin{lemma} \label{lem:plaus} Let $\bra{k_1},\bra{k_2}, \dots, \bra{k_s}$, where $0 \leq k_i \leq n$, be an $s$-tuple of $n$-tuples such that
\begin{description}
    \item[Cases (1), (3), and (4)] $\sum _{i=1}^s k_i = rn$.
    \item[Case (2)] $\sum _{i=1}^s k_i \leq rn$.
\end{description} Then $(\ts\bra{k_1},\ts\bra{k_2}, \dots, \ts\bra{k_s}) \in Q$. 

Moreover, in Cases (1), (2), and (3) we have $\ts \bra{n-k} = 1-\ts \bra{k}$ for every $0 \leq k \leq n$. 
\end{lemma}

\begin{proof}
For the first part, form an $s\times rn$ matrix $M$ whose first row is $\bra{k_1}_{rn}$ and the $j$-th row is the ($\sum_{l=1}^{j-1} k_l$)-th cyclic shift of $\bra{k_j}_{rn}$ for $j\in\{2,\dots,s\}$. Note that each column of $M$ contains exactly 1 one in Cases (1), (3), (4) and at most 1 one in Case (2). Split this matrix into $r$-many $s\times n$ blocks $M^1,M^2,\dots,M^r$. Their sum $X=\sum_{j=1}^r M^j$ is an $s\times n$ zero-one matrix whose each column contains exactly $r$ ones in Cases (1), (3), (4) and at most $r$ ones in Case (2) --- this is seen from the row sum of $M$, which is $\bra{\sum_{i=1}^s k_i}$. Moreover, for all $j\in [s]$, the $j$-th row of $X$ is a cyclic shift of $\bra{k_j}$, therefore its $\ts$-image is $\ts \bra{k_j}$  by \cref{tscyclic}. Each column belongs to the relation $P$, therefore, as $t^\sigma$ is a polymorphism, we get that $\ts$ applied to the rows gives a tuple in $Q$. This implies the first claim.

For the second part, we take $\bra{k}$ together with the $k$-th cyclic shift of $\bra{n-k}$ and use the fact that $\ts$ preserves the disequality relation pair.
\end{proof}

The next lemma is the technical core of this subsection.

\begin{lemma} \label{lem:stepone}
Denote $a = \lfloor \tresh n  \rfloor$.
For every $0 \leq k \leq 2a$, we have 
\[
 \ts \bran{k} = \left\{ \begin{array}{ll} 
 \ts \bran{0} & \mbox{if $0 \leq k  \leq a$} \\ 
 1-\ts \bran{0} & \mbox{if $1+a \leq k  \leq 2a$} \\ 
 \end{array}\right.
\]
\end{lemma}

\begin{proof}
\textbf{Case (1).} Here $a= \lfloor n/2 \rfloor = (n-1)/2$. 
We prove $\ts \bra{k} = 0$ and $\ts \bra{n-k} = 1$ for any $0 \leq k \leq a$ by induction on $i = a - k$, $i = 0, 1, \dots, a$. For the first step, $k = (n-1)/2$, we apply \cref{lem:plaus} to the $s$-tuple  $2r \times \bra{k}, \bra{r}, (s-2r-1) \times \bra{0}$. (Note that $r<s/2$, so $s-2r-1 \geq 0$, and that we can apply \cref{lem:plaus} because $2rk+r=rn$.) Since $Q$, which is $\ltrins$, contains no tuple with more than $(2r-1)$ ones, we get $\ts \bra{k} = 0$. Then also $\ts \bra{n-k}=1$ by the second part of the lemma. For the induction step, we use
the tuple 
\[
r \times \bra{k}, r \times \bra{n-k-1}, \bra{r}, (s-2r-1) \times \bra{0}
\] 
in a similar way, additionally using that $\ts \bra{n-k-1} = 1$ by the induction hypothesis. 

\textbf{Case (2).}
    If $0 \leq k \leq a$, then we apply \cref{lem:plaus} to the $s$-tuple $\bra{k}, \bra{k}, \dots, \bra{k}$ (we can do that because $sk \leq sa = s\lfloor n/2 \rfloor \leq sn/2 = rn$) and we get that the tuple $(\ts\bra{k},\ts\bra{k},\dots, \ts\bra{k})$ is in $\ltrins$; therefore $\ts\bra{k}=0$. For the remaining values $2a \geq k \geq a+1$ we apply the second part of the same lemma and get $\ts \bra{k}=1$. 

\textbf{Case (3).}   Here $a = \lfloor n/2 \rfloor = \lfloor rn/s \rfloor$.
    As $n \equiv 1 \pmod{s}$ and $r < s$, it follows that $a = r(n-1)/s$. 
    We will prove, starting from the left, the following chain of disequalities.
\[
\ts \bra{a} \neq \ts \bra{a+1} \neq \ts \bra{a-1} \neq \ts \bra{a+2} \neq \ts \bra{a-2} \neq \dots \neq \ts \bra{2a} \neq \ts \bra{0}
\]
This will imply $\ts \bra{a}= \ts \bra{a-1} = \dots = \ts \bra{0} \neq \ts \bra{a+1} = \ts \bra{a+2} = \dots = \ts \bra{2a}$. 
We start with the first disequality $\ts \bra{a} \neq \ts \bra{a+1}$. We apply \cref{lem:plaus} to the $s$-tuple
\[
(s-r) \times \bra{a}, \ r \times \bra{a+1},
\]
which we can as $(s-r)a + r(a+1) = sa+r = sr(n-1)/s + r = rn$.
The $\ts$ image of this tuple belongs to $\ltrins$, so $\ts \bra{a}$ and $\ts \bra{a+1}$ are not both ones. We also apply \cref{lem:plaus} to the $s$-tuple of ``complementary'' tuples 
\[
(s-r) \times \bra{n-a}, \ r \times \bra{n-(a+1)}
\]
and get that $\ts \bra{n-a}$ and $\ts \bra{n-(a+1)}$ are not both ones. From the second part of \cref{lem:plaus} it follows that $\ts \bra{a}$ and $\ts \bra{a+1}$ are not both zeros. We conclude that $\ts \bra{a} \neq \ts \bra{a+1}$.

For the second disequality $\ts \bra{a+1} \neq \ts\bra{a-1}$, we use the sequence
\[
(s-r)/2 \times \bra{a-1}, \ (s+r)/2 \times \bra{a+1}
\]
and derive $\ts \bra{a+1} \neq \ts \bra{a-1}$ using \cref{lem:plaus} in a similar way.

To prove  $\ts \bra{a-i+1} \neq \ts \bra{a+i}$ for $i \in\{2,3,\dots,a\}$, we
observe that, by the already established disequalities, we have $\ts \bra{a-i+1}= \dots = \ts \bra{a}$, 
and then use
\begin{itemize}
    \item $(s+r)/4 \times \bra{a+i}, \ (s-r)/2 \times \bra{a-1}, \ (s+r)/4 \times \bra{a-i+2}$ if $(s+r)/2$ is even; 
    \item $(s+r+2)/4 \times \bra{a+i}, \ (s-r-2)/2 \times \bra{a-1}, \  2 \times \bra{a-i+1}, \ (s+r-6)/4 \times \bra{a-i+2}$ if $(s+r)/2$ is odd.
\end{itemize}
Finally, for proving $\ts \bra{a+i} \neq\ts{\bra{a-i}}$ we use
\[ (s-r)/2 \times \bra{a-i}, \ (s-r)/2 \times \bra{a+i}, \ r \times \bra{a+1}
\]
This completes the proof for Case (3).

\textbf{Case (4).} The proof is similar to Case (3). We again have $a = r(n-1)/s$ and we prove the same sequence of disequalities.
The first disequality  $\ts \bra{a} \neq \ts \bra{a+1}$ is proved using the same sequence $(s-r) \times \bra{a}, \ r \times \bra{a+1}$. This time, however, we directly obtain $\ts \bra{a} \neq \ts \bra{a+1}$ since  $Q$ is the not-all-equal relation, so we do not need to use (and cannot use) the complementary tuples.
The proof proceeds in the same way as in Case (3) if $r$ and $s$ have the same parity. 

Suppose now that $r$ is even and $s$ is odd.
For the second disequality $\ts \bra{a+1} \neq \ts\bra{a-1}$,  we first use
\[
(s-1) \times \bra{a}, \ \bra{a+r}
\]
to deduce $\ts \bra{a+r} \neq \ts \bra{a}$ (so $\ts \bra{a+1} = \ts \bra{a+r}$) and then
\[
(s-1)/2 \times \bra{a-1}, \ (s-1)/2 \times \bra{a+1}, \ \bra{a+r}
\]
to deduce $\ts \bra{a-1} \neq \ts \bra{a+1}$. 

To prove  $\ts \bra{a-i+1} \neq \ts \bra{a+i}$ for $i \in\{2,3,\dots,a\}$,  we use
\[
r/2 \times \bra{a+i}, \ (s-r) \times \bra{a}, \ r/2 \times \bra{a-i+2}
\]
and to prove $\ts \bra{a+i} \neq\ts{\bra{a-i}}$ we use
\[
(s-1)/2 \times \bra{a-i}, \ (s-1)/2 \times \bra{a+i}, \ 1 \times \bra{a+r}
\]
The proof of \cref{lem:stepone} is concluded.
\end{proof}

The goal of this subsection is now an easy consequence.

\begin{lemma} \label{lem:stepsizeone}
Every near-threshold almost rectangle of step size at most one is tame. 
\end{lemma}
\begin{proof}
    Let $\vc{z}$ be a near-threshold almost rectangle of step size at most one. By using an appropriate $p$-ary cyclic shift we can, without loss of generality, assume that it is of the form $\vc{z} = \brap{z^2+1, \dots, z^2+1,z^2, \dots, z^2}$. Denoting $k$ the number of ones in $\vc{z}$, we have $t(\vc{z}) = \ts \bran{k}$. The claim now follows from \cref{lem:stepone} once we observe that $k \leq 2\lfloor \tresh n \rfloor$. Indeed, since $\vc{z}$ is a near-threshold almost rectangle, we have
    \[\left|\lambda(\vc{z})-\tresh\right| < \frac{1}{s^{\dz+3}},\]
    so 
    \[\frac{k}{n} = \lambda(\vc{z}) <\frac{1}{s^{\dz+3}}+\tresh \leq \frac{1}{s^3} + \tresh,\]
    and then $k \leq 2\lfloor \tresh n \rfloor$ follows for a sufficiently large $n$.
\end{proof}

\subsection{Arbitrary step size} \label{subsec:stepx}

The goal of this subsection is to prove that every near-threshold almost rectangle is tame.

The first lemma is a ``2-dimensional analogue'' of \cref{lem:plaus}. We require the following concept.

\begin{definition}
We say that an $m$-tuple of evaluations $\vc{k}_1 = \bra{k^1_1, k^2_1, \dots, k^p_1}$, $\vc{k}_2 = \bra{k^1_2, k^2_2, \dots, k^p_2}$, \dots, $\vc{k}_m = \bra{k^1_m, k^2_m, \dots,$ $k^p_m}$, where $m\in[s]$, is \emph{plausible} if $\sum _{j=1}^m k_j^i = rp$ for all $i\in [p]$.
\end{definition}

In other words, by arranging the integers that define $\vc{k}_1$, $\vc{k}_2$, \dots, $\vc{k}_m$ as rows of an $m \times p$ matrix, we get a matrix whose every column sums up to $rp$. Note that the sum of the areas of the evaluations is then equal to $r$.

\begin{lemma} \label{lem:pl}
If a tuple $\vc{k}_1,\vc{k}_2, \dots, \vc{k}_s$ is plausible, then $(t(\vc{k}_1),t(\vc{k}_2), \dots, t( \vc{k}_s)) \in Q$.

Moreover, in Cases (1), (2), and (3) we have $t \bra{p-k^1,p-k^2, \dots, p-k^p} = 1-t \bra{k^1,k^2, \dots, k^p}$ for any evaluation $\bra{k^1,k^2, \dots, k^p}$.
\end{lemma}

\begin{proof}
Let $\vc{k}_1,\vc{k}_2, \dots, \vc{k}_s$ be a plausible tuple.
Fix, for a while, an arbitrary $i\in [p]$. Form a $s\times rp$ matrix $M_i$ whose first row is $\bra{k_1^i}_{rp}$ and $j$-th row is the ($\sum_{l=1}^{j-1} k_l^i$)-th cyclic shift of $\bra{k_j^i}_{rp}$ for $j\in\{2,\dots,s\}$. Split this matrix into $r$-many $s\times p$ blocks $M_i^1,M_i^2,\dots,M_i^r$. Their sum $X_i=\sum_{j=1}^r M_i^j$ is an $s\times p$ matrix whose each column contains exactly $r$ ones. 
Moreover, for all $j\in [s]$, the $j$-th row of the matrix $X_i$ is a cyclic shift of $\bra{k_j^i}_p$.
Put the matrices $X_1$,$X_2$, \dots, $X_p$ aside to form an $s \times n$ matrix $Y$. Its rows 
have the same $t$-images as $\vc{k}_1,\vc{k}_2, \dots, \vc{k}_s$, respectively, because $t$ is doubly cyclic. Each column belongs to the relation $P$ (as it contains $\rins$), therefore, as $t$ is a polymorphism, we get that $t$ applied to the rows gives  a tuple in $Q$. This tuple is equal to $(t(\vc{k}_1),t(\vc{k}_2), \dots, t( \vc{k}_s))$. 

The second part can be proved in a similar way as the second part of \cref{lem:plaus} using the disequality relation pair.
\end{proof}

The next lemma will be applied to produce plausible sequences of evaluations. 

\begin{lemma} \label{lem:producing_plaus} Let $\vc{z}$ be an almost rectangle of step size $\dz \geq 2$ with  $|\lambda (\vc{z})-\tresh| \leq 1/s^3$. Then
\begin{itemize}
    \item there exists a plausible $r/\tresh$-tuple  $\vc{k}_1$, $\vc{k}_2$,\dots, $\vc{k}_{r/\tresh-1}$, $\vc{l}$ of almost rectangles such that $t(\vc{z}) = t(\vc{k}_1) = t(\vc{k}_2) = \dots = t(\vc{k}_{r/\tresh-1})$, $\lambda(\vc{z}) = \lambda(\vc{k}_1) = \dots = \lambda(\vc{k}_{r/\tresh-1})$, and $\vc{l}$ has the same step size $\dz$ as $\vc{z}$;
  \item there exists a plausible $r/\tresh$-tuple $\vc{k}_1$, $\vc{k}_2$,\dots, $\vc{k}_{r/\tresh-2}$, $\vc{l}_1, \vc{l}_2$ of almost rectangles such that $t(\vc{z}) = t(\vc{k}_1) = t(\vc{k}_2) = \dots = t(\vc{k}_{r/\tresh-2})$, $\lambda(\vc{z}) = \lambda(\vc{k}_1) = \lambda(\vc{k}_2) = \dots = \lambda(\vc{k}_{r/\tresh-2})$, both $\vc{l}_1$ and $\vc{l}_2$ have  step size strictly smaller than $\dz$, and $|\lambda(\vc{l}_1) - \lambda(\vc{l}_2)| \leq 1/p$.
  \end{itemize}
\end{lemma}

Before the proof, note that $r/\tresh = 2r$ in Cases (1), (2), (3)
and $r/\tresh=s$ in Cases (2), (3), (4).

\begin{proof}
Without loss of generality we can assume that $\vc{z}= \bra{c\times z^1,d\times z^2}$ for some $c,d$ and $z^1 > z^2$. Set $m = r/\tresh-1$ for the first item and $m = r/\tresh-2$ for the second one. 
We define an integer $m \times p$ matrix $X$ so that the first row is $(c \times z^1, d \times z^2)$ and the $i$-th row is the $c$-th cyclic shift of the $(i-1)$-st row for each $i \in \{2, \dots, m\}$.
Let $Y$ be the $(m+1) \times p$ matrix obtained from $X$ by adding a row $(l^1, l^2, \dots, l^p)$ so that each column sums up to $rp$.
It is easily seen by induction on $i \leq m$ that the sum of the first $i$ rows is a cyclic shift of a tuple of the form $(e, \dots, e, e', \dots, e')$, where $|e-e'| = \dz$ and the ``step down'' is at position $ci \mod p$ (when the columns are indexed from 0). It follows that $(l^1,l^2, \dots, l^p)$ is also a cyclic shift of a tuple of the form $(e, \dots, e, e', \dots, e')$ where $e$ and $e'$ differ by $\dz$.

Next we observe that each $l^i$ is greater than 0 if $p$ is sufficiently large. Indeed, note that since $|z^1-z^2|/p$ can be made arbitrarily small (recall $|z^1-z^2|<5b)$, we have $p(\lambda(\vc{z})-\epsilon)< z^1, z^2 < p(\lambda(\vc{z}) + \epsilon)$, where $\epsilon>0$ can be made arbitrarily small. Since $z^1, z^2 < p(\lambda(\vc{z}) + \epsilon)$ and for each $i$ we have $l^i=rp-mz^1$ or $l^i=rp-mz^2$, then we have, for each $i$,
\begin{align*}
l^i &> rp - mp(\lambda(\vc{z})+\epsilon) \\
& \geq rp - \left(\frac{r}{\tresh}-1\right)p\left(\tresh+\frac{1}{s^3}+\epsilon\right) \\ 
& = p\left(r-\left(\frac{r}{\tresh}-1\right) \tresh - \left(\frac{r}{\tresh}-1\right)\left(\frac{1}{s^3}+\epsilon\right)\right)\\
& = p\left(\tresh - \left(\frac{r}{\tresh}-1\right)\left(\frac{1}{s^3}+\epsilon\right)\right) \\ 
& > p\left(\tresh - \frac{r}{\tresh} \left(\frac{1}{s^3} + \epsilon\right)\right),
\end{align*}
which is, for a sufficiently small $\epsilon$, greater than 0 since $r/(\tresh s^3) \leq 1/s^2 < \tresh$. Similarly, now using $p(\lambda(\vc{z})-\epsilon)< z^1, z^2$, we get that each $l^i < 2\tresh p \leq p$ if $m = r/\tresh-1$ and $l^i < 3 \tresh p \leq 3p/2$ if $m = r/\tresh -2$.

Now we can finish the proof of the first item. We set
 $\vc{k}_1, \vc{k}_2, \dots, \vc{k}_m, \vc{l}$ to be the $n$-tuples determined by the rows of $Y$ via $\bra{}$, e.g., $\vc{l} = \bra{l^1,l^2,\dots, l^p}$. The inequalities $0 \leq l^i \leq p$ guarantee that $\vc{l}$ is correctly defined and we see, using also the double cyclicity of $t$ (for $t(\vc{z}) = t(\vc{k}_1) = t(\vc{k}_2) = \dots = t(\vc{k}_m)$), that these $n$-tuples have all the required properties.

To finish the proof of the second item, we define the $\vc{k}_i$ as above and set $\vc{l}_1 = \bra{\lfloor l^1/2 \rfloor, \lfloor l^2/2 \rfloor, \dots,\lfloor l^p/2 \rfloor}$, $\vc{l}_2 = \bra{\lceil l^1/2 \rceil,\lceil l^2/2 \rceil, \dots, \lceil l^p/2 \rceil}. $
Since $0 \leq l^i \leq 3p/2$, for $i=1,2,\dots,p$, these tuples are correctly defined almost rectangles. Their areas clearly differ by at most $1/p$. Since $\dz \geq 2$, their step sizes are strictly smaller than $\dz$, and we are done in this case as well.
\end{proof}

Equipped with these lemmata we are ready to reach the goal of this subsection. 

\begin{lemma} \label{lem:main}
Every near-threshold almost rectangle is tame. 
\end{lemma}
\begin{proof}
The proof is by induction on the step size. Step sizes zero and one are dealt with in \cref{lem:stepsizeone}, so we assume that $\vc{z}$
is a near-threshold almost rectangle of step size $2 \leq \dz < 5\bbb$.

Assume first that $\lambda(\vc{z})$ is not \emph{too close to $\tresh$}, namely, $|\lambda(\vc{z})-\tresh| \geq 1/s^{5\bbb+4}$. We apply the second item in \cref{lem:producing_plaus} and get a plausible $r/\tresh$-tuple $\vc{k}_1,\vc{k}_2$, \dots, $\vc{k}_{r/\tresh-2}, \vc{l}_1, \vc{l}_2$ such that $\vc{z}, \vc{k}_1, \vc{k}_2$, \dots, $ \vc{k}_{r/\tresh-2}$ all have the same $t$-images and areas, and  $\vc{l}_1$ and $\vc{l}_2$ are almost rectangles with step sizes strictly smaller than $\dz$, whose areas differ by at most $1/p$. 

The average area of almost rectangles $\vc{k}_1,\vc{k}_2$, \dots, $\vc{k}_{r/\tresh-2}$, $\vc{l}_1$, $\vc{l}_2$ is $\tresh$, the first $r/\tresh-2$ of them have the same area as $\vc{z}$, bounded away from $\tresh$ by a constant (namely $1/s^{5b+4}$), and the last two have almost the same area (the difference is at most $1/p$). By choosing a large enough $p$ we get 
\begin{equation*}
\sgn\left(\lambda(\vc{l}_1)-\tresh\right) = \sgn\left(\lambda(\vc{l}_2)-\tresh\right) \neq \sgn\left(\lambda(\vc{z})-\tresh\right)
\end{equation*}
and 
\begin{equation*}
\left|\lambda(\vc{l}_i)-\tresh\right| \leq \frac{r}{\tresh} \cdot \left|\lambda(\vc{z})-\tresh\right|;
\end{equation*}
in particular, both $\vc{l}_i$ are near-threshold since 
\begin{equation*}
\frac{r}{\tresh} \cdot \left|\lambda(\vc{z})-\tresh\right|\leq s \cdot \left|\lambda(\vc{z})-\tresh\right| \leq \frac{1}{s^{\dz + 3 -1}} \leq \frac{1}{s^{\Delta l_i +3}}.
\end{equation*}
By the induction hypothesis, both $\vc{l}_i$ are tame. Since $\sgn(\lambda(\vc{z})-1/2) \neq \sgn(\lambda(\vc{l}_1)-1/2)$ and
$t(\vc{z}) = t(\vc{k}_1) = t(\vc{k}_2) = \dots = t(\vc{k}_{r/\tresh-2})$,  $t(\vc{l}_1) = t(\vc{l}_2)$, it is now enough to show that 
the values $t(\vc{k}_1)$, $t(\vc{k}_2)$, \dots, $t(\vc{k}_{r/\tresh-2})$, $t(\vc{l}_1)$, $t(\vc{l}_2)$ are not all equal. 
In Case (4), this follows directly from \cref{lem:pl} as we have $r/\tresh = s$ tuples and $Q$ is the not-all-equal relation.
In Cases (2) and (3), we also have $s$ tuples but we can only derive that the values are not all one. In these cases we also apply \cref{lem:pl} to the ``complementary'' tuples  $\vc{\overline{k}}_1$, $\vc{\overline{k}}_2$, \dots, $\vc{\overline{k}}_{s-2}$, $\vc{\overline{l}}_1$, $\vc{\overline{l}}_2$, where for $\vc{k}=\bra{k^1,k^2, \dots, k^p}$, we define $\vc{\overline{k}}=\bra{p-k^1,p-k^2, \dots, p-k^p}$. We get that the $t$-images cannot all be one, so the values 
$t(\vc{k}_1),t(\vc{k}_2), \dots, t(\vc{k}_{s-2}), t(\vc{l}_1),t(\vc{l}_2)$ are not all zero by the second part of \cref{lem:pl}. 
Case (1) is proved in a similar way, we just need to complete the $r/\tresh = 2r$ tuples  to $s$ tuples by adding $s-2r$ zeros before applying \cref{lem:pl}.

It remains to deal with the case that $\lambda(\vc{z})$ is too close to $\tresh$, that is, $|\lambda(\vc{z})-\tresh| < 1/s^{5\bbb+4}$.
In this case we will shortly find an almost rectangle $\vc{l}$ with the same step size as $\vc{z}$ such that $t(\vc{l}) = 1- t(\vc{z})$ and $\lambda(\vc{l}) - \tresh = -s' (\lambda(\vc{z})-\tresh)$, where $s'$ is such that $2 \leq s' \leq s$. Having such an $\vc{l}$, if $\lambda(\vc{l})$ is already not too close to $\tresh$, then we observe that $\vc{l}$ is near-threshold (indeed,  $|\lambda(\vc{l})-\tresh| \leq s|\lambda(\vc{z})-\tresh| \leq s/s^{5b+4} < 1/s^{\Delta l+3}$)
and apply to $\vc{l}$ the first part of the proof, thus obtaining that $\vc{l}$ is tame and, consequently, $\vc{z}$ is tame as well. If $\lambda(\vc{l})$ is still too close to $\tresh$, then we simply repeat the process until we get a rectangle that is not too close.  

It remains to find such an almost rectangle $\vc{l}$. We apply the first item of \cref{lem:producing_plaus} and get a plausible $r/\tresh$-tuple $\vc{k}_1, \dots, \vc{k}_{r/\tresh-1}, \vc{l}$ such that $t(\vc{z}) = t(\vc{k}_1) = \dots = t(\vc{k}_{r/\tresh-1})$ and $\vc{l}$ is an almost rectangle of the same step size as $\vc{z}$. Since the area of each $\vc{k}_i$ is equal to $\lambda(\vc{z})$ and the average area in the plausible $r/\tresh$-tuple is $\tresh$, we get that $\lambda(\vc{l})-\tresh = - (r/\tresh-1)(\lambda(\vc{z})-\tresh)$. By the same trick as previously, using \cref{lem:pl}, using complementary tuples in Cases (2),(3), and additionally adding zeros in Case (1), we get that $t(\vc{l})$ and $t(\vc{z})$ are not equal. This concludes the construction of $\vc{l}$ and the proof of the lemma.
\end{proof}

\subsection{Contradiction} \label{subsec:contra}

The proof can now be finished by using the tameness of near-threshold almost rectangles together with the $\bbb$-boundedness of $t$ as follows. 

Let $m = (p-1)/2$ and choose positive integers $z^{2,1}$ and $z^{2,2}$ so that $\tresh p -2\bbb<z^{2,1}<z^{2,2} < \tresh p$,
and the x/y-tuples $(z^{2,1} \times x, (p-z^{2,1}) \times y)$  and 
$(z^{2,2} \times x, (p-z^{2,2}) \times y)$
are $\sim$-equivalent (see \cref{def:bounded} of boundedness). 
This is possible by the pigeonhole principle since there are more than $\bbb$ integers in the interval and $\sim$ has at most $\bbb$ classes. 

By the choice of $z^{2,1}$ and $z^{2,2}$, for any meaningful choice of $z^1$, we have
$t(\vc{z}_1) = t(\vc{z}_2)$ where $\vc{z}_i= \brap{m \times z^1, (p-m) \times z^{2,i}}$, $i=1,2$.  We choose $z^1$ as the maximum number such that $\lambda(\vc{z}_1) < \tresh$. Note that for $z^1=p$ the area of $\vc{z}_1$ can be made arbitrarily close to $(1+\tresh)/2 > \tresh$ by choosing a sufficiently large $p$, so we may assume $z^1<p$. 

From $m < p/2$ and the definition of $\vc{z}_i$ it follows that increasing $z^{2,1}$ by one makes the area of $\vc{z}_1$ greater than increasing $z^1$ by one, therefore  $\lambda(\vc{z}_2) > \tresh$.  

Note that $z^1 > \tresh p$ since otherwise the area of $\vc{z}_2$  is less than $\tresh$.  On the other hand,  $z^1 < \tresh p +3b$, otherwise the area of $\vc{z}_1$ is greater (assuming $p>5$):
\[
\lambda(\vc{z}_1)=\frac{mz^1+(p-m)z^{2,1}}{p^2} \geq \frac{\frac{p-1}{2}(\tresh p +3b)+\frac{p+1}{2}(\tresh p -2b)}{p^2} = \frac{ \tresh p^2 +\frac{b(p-5)}{2}}{p^2} > \tresh.
\]
It follows that the step size of both $\vc{z}_1$ and $\vc{z}_2$ is less than $5b$, so both $\vc{z}_i$ are almost rectangles. By choosing a sufficiently large $p$, the difference $\lambda(\vc{z}_2)-\lambda(\vc{z}_1)$ can be made arbitrarily small, and since $\lambda(\vc{z}_1) < \tresh < \lambda(\vc{z}_2)$ both $\vc{z}_i$ are then near-threshold.

Now the tameness of near-threshold almost rectangles proved in \cref{lem:main}
gives us $t(\vc{z}_1) = t\bran{0} \neq 1-t\bran{0} = t(\vc{z}_2)$. On the other hand, we also have $t(\vc{z}_1) = t(\vc{z}_2)$, a contradiction.

\section{Basic Cases}
\label{append}

In \cref{subsection:basiccases} we mentioned that tractable PCSPs from \cref{thm:bg} are, in a certain sense, building blocks for all tractable symmetric Boolean PCSPs allowing negations. Although this claim essentially follows from \cite{BG21}, we provide a proof for completeness.

The complexity classification is based on three important types of polymorphisms defined for a positive odd integer $n$ and $\vc{x}= (x_1, x_2, \dots, x_n) \in \{0,1\}^n$ as follows.
\begin{itemize}
    \item The Parity function: $\Par_n(\vc{x})=1$ iff $\Sigma_{i=1}^n x_i$ \text{ is odd}.
    \item The Majority function: $\Maj_n(\vc{x})=1$ iff $\Sigma_{i=1}^n x_i>n/2$.
    \item The Alternating-Threshold function: $\AT_n(\vc{x})=1$ iff $\Sigma_{i=1}^n (-1)^{i-1} x_i>0$.
\end{itemize}
The negations of these functions are denoted using a horizontal bar, e.g., $\overline{\Par}_n(\vc{x})= 1 - \Par(\vc{x})$. 

\begin{theorem}[\cite{BG21}] \label{1thm:bg}
Let $\Gamma$ be a symmetric Boolean $\PCSP$ template allowing negations. If at least one of $\Par_n$, $\Maj_n$, $\AT_n$, $\overline{\Par}_n$, $\overline{\Maj}_n$, $\overline{\AT}_n$ is a polymorphism of $\Gamma$ for all odd $n$, then $\PCSP(\Gamma)$ is polynomial-time solvable. Otherwise, $\PCSP(\Gamma)$ is NP-hard.
\end{theorem}

The following two lemmas are proved in \cite{BG21}
(Claim 4.6. and Claim 4.8.).

\begin{lemma}\label{1:at} Let $(P,Q)$ be a pair of Boolean symmetric relations of arity $s \geq 2$ such that $\AT_n$ is its polymorphism for every odd $n$. 
\begin{itemize}
    \item If $P$ contains $\bra{r}$ for some $r \in\{1,2,\dots,s-1\}$, then $Q$ contains $\naes$.
    \item If $P$ contains $\bra{r_1}$ and $\bra{r_2}$ for two different elements $r_1,r_2$  such that $\{r_1,r_2\}\neq \{0,s\}$, then $Q =\{0,1\}^s$.
\end{itemize}
\end{lemma}

\begin{lemma}\label{1:maj}
Let $(P,Q)$ be a pair of Boolean symmetric relations of arity $s \geq 2$ such that $\Maj_n$ is its polymorphism for every odd $n$. Suppose $\{r \mid \bra{r} \in P\}$ in not contained in $\{0,s\}$, and let $r_1$ and $r_2$ be the minimum and maximum of this set, respectively. 
    \begin{itemize}
    \item If $r_1 < s/2$, then $Q$ contains $\leq\!(2r_2-1)\text{-in-}s$.
    \item If $r_2 > s/2$, then $Q$ contains $\geq\!(2r_1-s+1)\text{-in-}s$.
    \item If $r_1=r_2=s/2$, then $Q$ contains $\naes$.
    \end{itemize}
\end{lemma}

The following lemma states an analogous fact for Parity. We provide a proof, since it is not explicit in \cite{BG21}. 

\begin{lemma}\label{1:par}
Let $(P,Q)$ be a pair of Boolean symmetric relations of arity $s \geq 2$ such that $\Par_n$ is its polymorphism for every odd $n$. 
    \begin{itemize}
        \item If $P$ contains $\bra{r}$ for some odd $r \in \{1,2, \dots, s-1\}$, then $Q$ contains $\oddins$.
        \item If $P$ contains $\bra{r}$ for some even $r \in \{1,2, \dots, s-1\}$, then $Q$ contains $\evenins$. 
    \end{itemize}
\end{lemma}
\begin{proof}
Let $q \in \{0,1, \dots, s\}$ be of the same parity as $r$. Our aim is to show that $\bra{q} \in Q$, which is enough as $Q$ is symmetric. If $q \geq r$, we set $n = q-r+1$ and use the following zero-one $s \times n$ matrix: the first $r-1$ rows are all ones, the next $n$ rows form the identity matrix, and the remaining rows are all zeros. Note that each column contains exactly $r$ ones, so it belongs to $P$. Applying $\Par_n$ we get $\bra{r-1+q-r+1} = \bra{q} \in Q$.
For $q < r$, we set $n = r-q+1$ and use the following matrix: the first $q$ rows are all ones, the next $n$ rows form the negation of the identity matrix, and the remaining rows are all zeros. Each column again contains $q+n-1 = r$ ones and applying $\Par_n$ gives us $\bra{q} \in Q$.
\end{proof}

The goal of this section is a simple consequence of these three lemmas.

\begin{theorem}\label{1thm:bsc}
    Every tractable symmetric Boolean PCSP allowing negations can be obtained by
\begin{itemize}
    \item taking 
    any number of relation pairs from one of the following three items (where $r$ and $s$ are positive integers):
    \begin{enumerate}[(a)]
\item $(\oddins,\oddins)$, or $(\evenins,\evenins)$
\item $(\lrins,\ltrins)$ and $r \leq s/2$, or \\ 
      $(\grins, \gtrins)$ and $r \geq s/2$, or \\
      $(\frac{s}{2}\mbox{-in-}s,\naes)$ and $s$ is even
\item $(\rins,\naes)$
\end{enumerate} 
    \item adding any number of ``trivial'' relation pairs $(P,Q)$ such that $P \subseteq Q$, and $Q$ is the full relation or $P$ contains only constant tuples, and 
    \item taking a homomorphic relaxation of the obtained template.
\end{itemize}
\end{theorem}

\begin{proof}
    Let $\Gamma$ be a symmetric Boolean $\PCSP$ template allowing negations such that $\PCSP(\Gamma)$ is solvable in polynomial time. Then by \cref{1thm:bg} (as we assumed P $\neq$ NP) we have that at least one of $\AT_n$, $\Maj_n$, $\Par_n$, $\overline{\AT}_n$, $\overline{\Maj}_n$, $\overline{\Par}_n$ is a polymorphism of $\Gamma$ for all odd $n$.
    
    In the first case, \cref{1:at} implies that in every nontrivial relation pair $(P,Q)$ in $\Gamma$, $P$ is contained in (in fact equal to) $\rins$ for some $r$, $s$ and $Q$ contains $\naes$. Therefore $\Gamma$ is a homomorphic relaxation of a template obtained by taking relation pairs from item (c) and trivial relation pairs, where the homomorphisms witnessing the relaxation are the identity functions.

     \cref{1:maj} implies the theorem for $\Maj_n$, item (b). Indeed, if $(P,Q)$ is nontrivial and $r_1, r_2$ are as in the lemma, then necessarily $r_2 \leq s/2$ or $s/2 \leq r_1$, otherwise $Q = \{0,1\}^s$ by the lemma.  If $r_1=r_2=s/2$ we get $P = \frac{s}{2}\text{-in-}s$ and $Q$ contains $\naes$. Otherwise, if $r_2 \leq s/2$, then $P$ is contained in $\leq\!r_2\text{-in-}s$ and $Q$ contains $\leq\!(2r_2-1)\text{-in-}s$, and if $s/2 \leq r_1$, then $P$ is contained in $\geq\!r_1\text{-in-}s$ and $Q$ contains $\geq\!(2r_1-s+1)\text{-in-}s$. 
     
     Similarly, \cref{1:par} implies the theorem for $\Par_n$, item (a).
     
    For $\overline{\AT}_n$, we consider a new template $\Gamma'$ obtained by negating, in each relational pair $(P,Q)$, all tuples in $Q$. The template $\Gamma'$ has $\AT_n$ as a polymorphism, so we can apply the previous argument to it. The original $\Gamma$ is then a homomorphic relaxation of the same template as $\Gamma'$, where the first homomorphism is the identity and the second one is the negation function. The proof for $\overline{\Maj}_n$, $\overline{\Par}_n$ is  analogous. 
\end{proof}

\section{Conclusion}

We have characterized finite tractability among the basic tractable cases in the Brakensiek--Guruswami classification~\cite{BG21} of symmetric Boolean PCSPs allowing negations. A natural direction for future research is an extension to all the tractable cases (not just the basic ones), or even to all symmetric Boolean PCSPs~\cite{FKOS19}, not only those allowing negations. An obstacle, where our efforts have failed so far, is already in relaxations of the basic templates $(P,Q)$ with disequalities. For example, which $(P,Q)$, $(\neq,\neq)$, with $P$ a subset of  $\lrins$ and $Q$ a superset of $\ltrins$, give rise to finitely tractable PCSPs?

Another natural direction is to better understand the ``level of tractability.'' For the finitely tractable templates $(\rel A,\rel B)$ considered in this paper, it is always possible to find a tractable $\CSP(\rel C)$ with $\rel A \to \rel C \to \rel B$ and such that $\rel C$ is two-element. Is it so for all symmetric Boolean templates? For general Boolean templates, the answer is ``No'': \cite{DSMMNS21} presented an example that requires a three-element $\rel C$ and later \cite{KMZ22} showed that there is no finite upper bound on the size of $\rel C$.  There are also natural concepts beyond finite tractability when some weaker finiteness assumption is placed on $\rel C$. Recent papers \cite{M25prep}, \cite{PRSS25prep} provide significant contributions in this direction.

\ifarxiv
\bibliographystyle{plainurl}
\fi
\iftcs
\bibliographystyle{elsarticle-num} 
\fi
\bibliography{ref}

\begin{thebibliography}{10}

\bibitem{AB21}
Kristina Asimi and Libor Barto.
\newblock {Finitely Tractable Promise Constraint Satisfaction Problems}.
\newblock In Filippo Bonchi and Simon~J. Puglisi, editors, {\em 46th International Symposium on Mathematical Foundations of Computer Science (MFCS 2021)}, volume 202 of {\em Leibniz International Proceedings in Informatics (LIPIcs)}, pages 11:1--11:16, Dagstuhl, Germany, 2021. Schloss Dagstuhl -- Leibniz-Zentrum f{\"u}r Informatik.
\newblock \href {https://doi.org/10.4230/LIPIcs.MFCS.2021.11} {\path{doi:10.4230/LIPIcs.MFCS.2021.11}}.

\bibitem{AGH17}
Per Austrin, Venkatesan Guruswami, and Johan H{\aa}stad.
\newblock \((2+{\epsilon})\)-{S}at is {NP}-hard.
\newblock {\em {SIAM} J. Comput.}, 46(5):1554--1573, 2017.
\newblock \href {https://doi.org/10.1137/15M1006507} {\path{doi:10.1137/15M1006507}}.

\bibitem{Bar19}
Libor Barto.
\newblock Promises make finite (constraint satisfaction) problems infinitary.
\newblock In {\em 2019 34th Annual ACM/IEEE Symposium on Logic in Computer Science (LICS)}, pages 1--8, 2019.
\newblock \href {https://doi.org/10.1109/LICS.2019.8785671} {\path{doi:10.1109/LICS.2019.8785671}}.

\bibitem{BBKO21}
Libor Barto, Jakub Bul\'in, Andrei Krokhin, and Jakub Opr{\v s}al.
\newblock Algebraic approach to promise constraint satisfaction.
\newblock {\em J. ACM}, 68(4):Art. 28, 66, 2021.
\newblock \href {https://doi.org/10.1145/3457606} {\path{doi:10.1145/3457606}}.

\bibitem{BK12}
Libor Barto and Marcin Kozik.
\newblock Absorbing subalgebras, cyclic terms, and the constraint satisfaction problem.
\newblock {\em Log. Methods Comput. Sci.}, 8(1:07):1--26, 2012.
\newblock Special issue: Selected papers of the Conference ``Logic in Computer Science (LICS) 2010''.
\newblock \href {https://doi.org/10.2168/LMCS-8(1:07)2012} {\path{doi:10.2168/LMCS-8(1:07)2012}}.

\bibitem{BKW17}
Libor Barto, Andrei Krokhin, and Ross Willard.
\newblock Polymorphisms, and how to use them.
\newblock In Andrei Krokhin and Stanislav {\v Z}ivn{\' y}, editors, {\em The Constraint Satisfaction Problem: Complexity and Approximability}, volume~7 of {\em Dagstuhl Follow-Ups}, pages 1--44. Schloss Dagstuhl--Leibniz-Zentrum fuer Informatik, Dagstuhl, Germany, 2017.
\newblock \href {https://doi.org/10.4230/DFU.Vol7.15301.1} {\path{doi:10.4230/DFU.Vol7.15301.1}}.

\bibitem{BOP18}
Libor Barto, Jakub Opr{\v{s}}al, and Michael Pinsker.
\newblock The wonderland of reflections.
\newblock {\em Israel Journal of Mathematics}, 223(1):363--398, Feb 2018.
\newblock \href {https://doi.org/10.1007/s11856-017-1621-9} {\path{doi:10.1007/s11856-017-1621-9}}.

\bibitem{Bod21}
Manuel Bodirsky.
\newblock {\em Complexity of infinite-domain constraint satisfaction}, volume~52 of {\em Lecture Notes in Logic}.
\newblock Cambridge University Press, Cambridge; Association for Symbolic Logic, Ithaca, NY, 2021.
\newblock \href {https://doi.org/10.1017/9781107337534} {\path{doi:10.1017/9781107337534}}.

\bibitem{BG21}
Joshua Brakensiek and Venkatesan Guruswami.
\newblock Promise constraint satisfaction: Algebraic structure and a symmetric boolean dichotomy.
\newblock {\em SIAM Journal on Computing}, 50(6):1663--1700, 2021.
\newblock \href {https://doi.org/10.1137/19M128212X} {\path{doi:10.1137/19M128212X}}.

\bibitem{BJK05}
Andrei Bulatov, Peter Jeavons, and Andrei Krokhin.
\newblock Classifying the complexity of constraints using finite algebras.
\newblock {\em SIAM J. Comput.}, 34(3):720--742, March 2005.
\newblock \href {https://doi.org/10.1137/S0097539700376676} {\path{doi:10.1137/S0097539700376676}}.

\bibitem{Bul17}
Andrei~A. Bulatov.
\newblock A dichotomy theorem for nonuniform {CSP}s.
\newblock In {\em 2017 IEEE 58th Annual Symposium on Foundations of Computer Science (FOCS)}, pages 319--330, October 2017.
\newblock \href {https://doi.org/10.1109/FOCS.2017.37} {\path{doi:10.1109/FOCS.2017.37}}.

\bibitem{DSMMNS21}
Guofeng Deng, Ezzeddine~El Sai, Trevor Manders, Peter Mayr, Poramate Nakkirt, and Athena Sparks.
\newblock Sandwiches for promise constraint satisfaction.
\newblock {\em Algebra Universalis}, 82(1):Paper No. 15, 8, 2021.
\newblock \href {https://doi.org/10.1007/s00012-020-00702-5} {\path{doi:10.1007/s00012-020-00702-5}}.

\bibitem{FV98}
Tom{\'{a}}s Feder and Moshe~Y. Vardi.
\newblock The computational structure of monotone monadic {SNP} and constraint satisfaction: A study through datalog and group theory.
\newblock {\em SIAM J. Comput.}, 28(1):57--104, February 1998.
\newblock \href {https://doi.org/10.1137/S0097539794266766} {\path{doi:10.1137/S0097539794266766}}.

\bibitem{FKOS19}
Miron Ficak, Marcin Kozik, Miroslav Ol{\v s}{\'a}k, and Szymon Stankiewicz.
\newblock {Dichotomy for Symmetric Boolean PCSPs}.
\newblock In Christel Baier, Ioannis Chatzigiannakis, Paola Flocchini, and Stefano Leonardi, editors, {\em 46th International Colloquium on Automata, Languages, and Programming (ICALP 2019)}, volume 132 of {\em Leibniz International Proceedings in Informatics (LIPIcs)}, pages 57:1--57:12, Dagstuhl, Germany, 2019. Schloss Dagstuhl--Leibniz-Zentrum fuer Informatik.
\newblock \href {https://doi.org/10.4230/LIPIcs.ICALP.2019.57} {\path{doi:10.4230/LIPIcs.ICALP.2019.57}}.

\bibitem{HN90}
Pavol Hell and Jaroslav Ne{\v{s}}et{\v{r}}il.
\newblock On the complexity of {$H$}-coloring.
\newblock {\em J. Combin. Theory Ser. B}, 48(1):92--110, 1990.
\newblock \href {https://doi.org/10.1016/0095-8956(90)90132-J} {\path{doi:10.1016/0095-8956(90)90132-J}}.

\bibitem{Jea98}
Peter Jeavons.
\newblock On the algebraic structure of combinatorial problems.
\newblock {\em Theor. Comput. Sci.}, 200(1-2):185--204, 1998.
\newblock \href {https://doi.org/10.1016/S0304-3975(97)00230-2} {\path{doi:10.1016/S0304-3975(97)00230-2}}.

\bibitem{KMZ22}
Alexandr Kazda, Peter Mayr, and Dmitriy Zhuk.
\newblock Small promise {CSP}s that reduce to large {CSP}s.
\newblock {\em Log. Methods Comput. Sci.}, 18(3):Paper No. 25, 14, 2022.
\newblock \href {https://doi.org/10.46298/lmcs-18(3:25)2022} {\path{doi:10.46298/lmcs-18(3:25)2022}}.

\bibitem{KKR17}
Vladimir Kolmogorov, Andrei Krokhin, and Michal Rol\'inek.
\newblock The complexity of general-valued {CSP}s.
\newblock {\em SIAM Journal on Computing}, 46(3):1087--1110, 2017.
\newblock \href {https://doi.org/10.1137/16M1091836} {\path{doi:10.1137/16M1091836}}.

\bibitem{KO22}
Andrei Krokhin and Jakub Opr\v{s}al.
\newblock An invitation to the promise constraint satisfaction problem.
\newblock {\em ACM SIGLOG News}, 9(3):30–59, August 2022.
\newblock \href {https://doi.org/10.1145/3559736.3559740} {\path{doi:10.1145/3559736.3559740}}.

\bibitem{LZ25}
Alberto Larrauri and Stanislav {\v Z}ivn{\'y}.
\newblock Solving promise equations over monoids and groups.
\newblock {\em ACM Trans. Comput. Log.}, 26(1):Art. 3, 24, 2025.
\newblock \href {https://doi.org/10.1145/3698106} {\path{doi:10.1145/3698106}}.

\bibitem{M24}
Antoine Mottet.
\newblock {Promise and Infinite-Domain Constraint Satisfaction}.
\newblock In Aniello Murano and Alexandra Silva, editors, {\em 32nd EACSL Annual Conference on Computer Science Logic (CSL 2024)}, volume 288 of {\em Leibniz International Proceedings in Informatics (LIPIcs)}, pages 41:1--41:19, Dagstuhl, Germany, 2024. Schloss Dagstuhl -- Leibniz-Zentrum f{\"u}r Informatik.
\newblock \href {https://doi.org/10.4230/LIPIcs.CSL.2024.41} {\path{doi:10.4230/LIPIcs.CSL.2024.41}}.

\bibitem{M25prep}
Antoine Mottet.
\newblock Algebraic and algorithmic synergies between promise and infinite-domain {CSPs}, 2025.
\newblock \href {https://arxiv.org/abs/2501.13740} {\path{arXiv:2501.13740}}.

\bibitem{NZ24}
Tamio-Vesa Nakajima and Stanislav {\v Z}ivn{\'y}.
\newblock On the complexity of symmetric vs. functional {PCSP}s.
\newblock {\em ACM Trans. Algorithms}, 20(4):Art. 33, 29, 2024.
\newblock \href {https://doi.org/10.1145/3673655} {\path{doi:10.1145/3673655}}.

\bibitem{PRSS25prep}
Michael Pinsker, Jakub Rydval, Moritz Schöbi, and Christoph Spiess.
\newblock Three fundamental questions in modern infinite-domain constraint satisfaction, 2025.
\newblock \href {https://arxiv.org/abs/2502.06621} {\path{arXiv:2502.06621}}.

\bibitem{Sch78}
Thomas~J. Schaefer.
\newblock The complexity of satisfiability problems.
\newblock In {\em Proceedings of the Tenth Annual ACM Symposium on Theory of Computing}, STOC '78, pages 216--226, New York, NY, USA, 1978. ACM.
\newblock \href {https://doi.org/10.1145/800133.804350} {\path{doi:10.1145/800133.804350}}.

\bibitem{Zhu17}
Dmitriy Zhuk.
\newblock A proof of {CSP} dichotomy conjecture.
\newblock In {\em 2017 IEEE 58th Annual Symposium on Foundations of Computer Science (FOCS)}, pages 331--342, Oct 2017.
\newblock \href {https://doi.org/10.1109/FOCS.2017.38} {\path{doi:10.1109/FOCS.2017.38}}.

\bibitem{Zhu20}
Dmitriy Zhuk.
\newblock A proof of the {CSP} dichotomy conjecture.
\newblock {\em J. ACM}, 67(5), August 2020.
\newblock \href {https://doi.org/10.1145/3402029} {\path{doi:10.1145/3402029}}.

\end{thebibliography}

\end{document}